\documentclass[conference]{IEEEtran}

\IEEEoverridecommandlockouts
\usepackage[letterpaper, left=0.75in, right=0.75in, bottom=0.75in, top=1in]{geometry}
\usepackage{cite}
\usepackage{amsmath,amssymb,amsfonts}
\usepackage{graphicx}
\usepackage{amsthm}
\usepackage[ruled,vlined]{algorithm2e}
\usepackage{textcomp}
\usepackage{xcolor}
\usepackage{caption}
\usepackage{subcaption}
\usepackage{bbold}
\usepackage{mathrsfs}
\usepackage[mathscr]{euscript}
\usepackage{cancel}
\def\BibTeX{{\rm B\kern-.05em{\sc i\kern-.025em b}\kern-.08em
    T\kern-.1667em\lower.7ex\hbox{E}\kern-.125emX}}

\usepackage{url}
\usepackage{tabularx}
\newtheorem{theorem}{Theorem}

\newtheorem{definition}{Definition}

\newtheorem{example}{Example}
\begin{document}

\title{\LARGE \textbf{Control Barrier Function Contracts for Vehicular Mission Planning under Signal Temporal Logic Specifications} %under Signal Temporal Logic Requirements}
\thanks{$^{1}$Muhammad Waqas (\texttt{waqas@usc.edu}), Nikhil Vijay Naik (\texttt{nikhilvn@usc.edu}),  Petros Ioannou (\texttt{ioannou@usc.edu}), and Pierluigi Nuzzo (\texttt{nuzzo@usc.edu}) are affiliated with the Ming Hsieh Department of Electrical and Computer Engineering, University of Southern California, Los Angeles. 
This research was supported in part by the National Science Foundation (NSF) under Awards 1839842, 1846524, and 2139982, the Office of Naval Research (ONR) under Award N00014-20-1-2258, and the Defense Advanced Research Projects Agency (DARPA) under Award HR00112010003. 
% and METRANS Transportation Center under the following grants: Pacific Southwest Region 9 University Transportation Center (USDOT/Caltrans) and the National Center for Sustainable Transportation (USDOT/Caltrans).
}
}

\author{Muhammad Waqas$^{1}$, Nikhil Vijay Naik$^{1}$, Petros Ioannou$^{1}$, and Pierluigi Nuzzo$^{1}$}
%\IEEEauthorblockA{\textit{Ming Hsieh Department of Electrical and Computer Engineering} \\
%\textit{University of Southern California}\\
%Los Angeles, USA \\
%nikhilvn@usc.edu, nuzzo@usc.edu}

\maketitle

\begin{abstract}
We present a compositional control synthesis method based on assume-guarantee contracts with application to correct-by-construction design of vehicular mission plans. In our approach, a mission-level specification expressed in a fragment of signal temporal logic (STL) is decomposed into formulas whose predicates are defined on {non-overlapping} time intervals. The STL formulas are then mapped to aggregations of contracts associated with continuously differentiable time-varying control barrier functions. The barrier functions are used to constrain the lower-level control synthesis problem, which is solved via quadratic programming. Our approach can mitigate the conservatism of previous methods for task-driven control based on under-approximations. We illustrate its effectiveness on a case study motivated by vehicular mission planning under safety constraints as well as constraints imposed by traffic regulations under vehicle-to-vehicle and vehicle-to-infrastructure communication. 
\end{abstract}

\section{Introduction}\label{sec:Intro}

The necessity of ensuring mission safety of autonomous cyber-physical systems such as vehicles immersed in an urban setting~\cite{shalev2017formal} has motivated the development of correct-by-construction, algorithmic control synthesis methods (see, e.g., \cite{ames2019control,tabuada2006linear}) % xiao2019decentralized, zeng2021safety} 
to  help ensure that a system fulfills its mission requirements while avoiding potentially hazardous configurations. 

A major challenge to control synthesis stems from the heterogeneity of formalisms needed to design and analyze complex cyber-physical systems~\cite{shoukry2018smc}. Some of the efforts in the literature leverage symbolic approaches to effectively synthesize provably correct high-level task planners. However, by relying on discrete abstractions of the design space, these methods may be prone to scalability issues when applied to complex continuous systems. On the other hand, low-level feedback control synthesis methods have shown to be effective in enforcing invariance and simple reachability properties on continuous systems. They have, however, difficulty in capturing more complicated mission constraints, including logical constraints, often inducing discontinuities in the target safe sets. 
More recently, the representation of the mission specification in an expressive logic language, such as signal temporal logic (STL), 
%rather than time-invariant safe sets
together with mixed integer linear encodings of the STL formulas~\cite{raman2014model} have been proposed to perform discrete-time trajectory planning in a model predictive control fashion for a wider class of objectives, including time-sensitive constraints. However, efficiently encompassing mission-level (logical) and control-level (dynamical) constraints within a unifying framework remains a challenge.

Compositional and hierarchical methods show the promise of harnessing the complexity due to the scale and heterogeneity of the control design problem, e.g., via a layered approach that can capture different kinds of constraints at different layers, without inducing excessive conservatism in the solutions.  
In this context, assume-guarantee (A/G) contracts have been  employed~\cite{nuzzo2013contract, nuzzo2015platform} to support 
compositional synthesis under temporal logic specifications. A/G reasoning has also been explored to argue about the correct composition of lane keeping and cruise control for vehicular planning~\cite{xu2017correctness}.   
% controllers catering to allied sub-requirements on the plant. 
However, an A/G contract framework that can effectively bridge high-level planning and continuous-time feedback control 
is an open research problem.

This paper addresses the above challenges by exploring a formalization of control barrier functions~\cite{ames2019control} in terms of A/G contracts capable of bridging high-level task planning and low-level feedback control. Central to our approach is the characterization of time-varying safe sets via a composition of continuously differentiable time-varying control barrier function ($\mathcal{C}^1$ TV-CBF) contracts that can capture time-varying constraints including jump discontinuities. Our contributions can be summarized as follows:
 
\begin{itemize}

\item We formalize a notion of \emph{time-varying finite-time convergence control barrier function} (TV-FCBF) as a contract providing an effective interface between task planning and feedback control synthesis. 

\item We determine necessary and sufficient conditions for the composition of TV-CBF contracts to generate a compatible contract, for which a controller is guaranteed to exist. 

\item By building on these abstractions, we introduce an algorithm that maps a mission-level specification expressed in a fragment of STL to an aggregation of CBF contracts from which a feedback controller can be  designed via quadratic programming. 
\end{itemize}

% Here the related work for people to appreciate the fine prints
Our method is reminiscent of previous approaches to STL control synthesis using 
% for continuous-time systems with time-varying
CBFs~\cite{lindemann2018control}, in that we associate candidate CBFs with atomic STL predicates in the specification. However, our approach can mitigate the potential conservatism induced by previous methods, based on concatenating multiple CBFs via a pointwise minimum operator and approximating the result via a smooth function, 
which may lead to overly defensive behaviors. 

Our synthesis algorithm is also inspired by funnel-based control synthesis~\cite{majumdar2017funnel}, where funnels associated with controllers from a predefined library are sequentially composed. 
% The funnel of forward-reachable sets is computed around a nominal trajectory. These funnels are then sequentially composed such that at the time of switching from one funnel to the other, the current funnel is the subset of the next funnel. 
\begin{figure}[t]
        \centering \includegraphics[width=7cm]{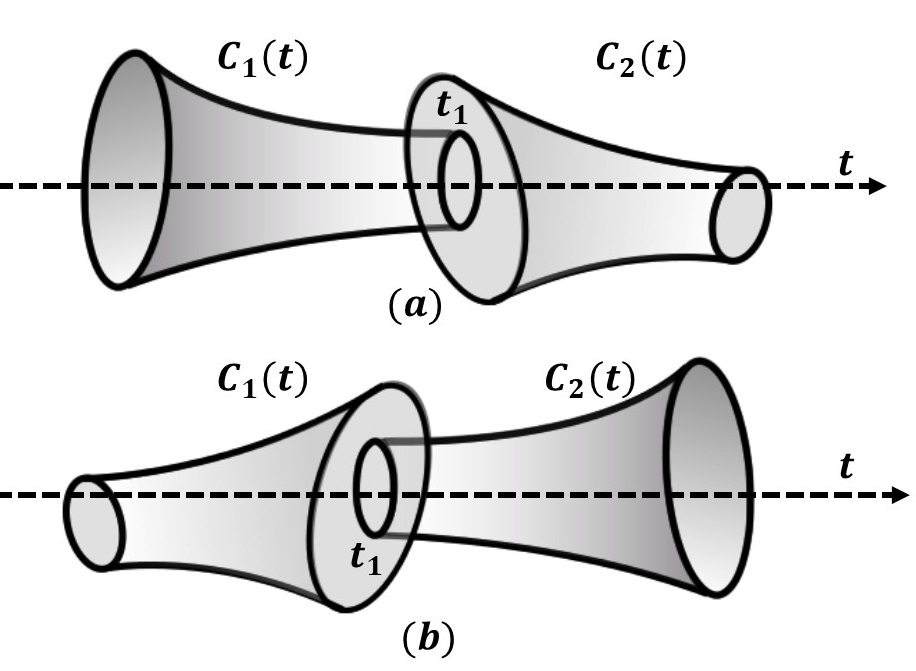}
          \caption{\small In a funnel-based approach  funnels are composable at time $t=t_1$ only under scenario (a) $C_1(t_1) \subseteq C_2(t_1)$. Our approach is based on the composition of time-varying safe sets expressed by contracts. Contracts compatibility requires $C_1(t_1) \cap C_2(t_1) \neq \emptyset$, as in scenario (b).} 
          \label{fig:funnel_analogy}
 \end{figure}
In a funnel-based approach, the current funnel is required to be a subset of the upcoming funnel at the time of switching (see Fig.~\ref{fig:funnel_analogy}). Our approach is, instead, based on the composition of time-varying safe sets, without \emph{a priori} constraining the architecture of the controller that will be  engaged. We can then relax the requirement that the current set be a subset of the upcoming one by relying on a time-varying version of a finite-time convergence CBF contract~\cite{li2018formally}. We illustrate the effectiveness of the proposed approach on a case study of vehicular motion planning under safety and regulatory constraints like traffic signals and variable speed limit under vehicle-to-vehicle (V2V) and vehicle-to-infrastructure (V2I) communication.

\section{Background}\label{sec:Background}

\subsection{A/G Contracts}\label{subsec:ContractBasedDesign}

A \emph{contract} $\mathscr{C}$ for a component $M$ is a triple $(V, A, G)$, where $V$ is a set of variables, and $A$ and $G$ are sets of behaviors over $V$. $A$, termed the \emph{assumptions}, encode the assumptions made by $M$ on its operational environment. $G$ is the set of \emph{guarantees}, i.e., the collection of behaviors promised by $M$ provided that the environment satisfies $A$. We say that $M$ \emph{satisfies} 
% the contract 
$\mathscr{C}$ when all the behaviors of ${M}$ satisfying  
% the assumptions 
${A}$ are contained in
% inside the guarantees 
${G}$. A contract is \emph{consistent} if there exists a valid implementation ${M}$, i.e., 
% given an assumption set $\mathrm{A}$, 
$G \cup \bar{A}$ is nonempty. It is said to be \emph{compatible}, if there exists a valid environment ${E}$, i.e., ${A}$ is nonempty. 
We can compare two contracts $\mathscr{C}_1$ and $\mathscr{C}_2$ through the \emph{refinement} operation, which is a preorder on contracts. We say that $\mathscr{C}_1$ refines $\mathscr{C}_2$ and write $\mathscr{C}_1 \preceq \mathscr{C}_2$ if and only if $\mathscr{C}_1$ has weaker assumptions and stronger guarantees. The \emph{conjunction} of two contracts $\mathscr{C}_1$ and $\mathscr{C}_2$ is defined as the contract serving as the greatest lower bound which refines both. This can be used to represent a combination of requirements that must be satisfied simultaneously. The \emph{composition} of contracts is, instead, used to derive a more complex contract that must be satisfied by a composition of components, each satisfying its local contract. 
A detailed exposition of all the terms summarized above may be found in the literature~\cite{benveniste2012contracts}. 

\subsection{Control Barrier Functions}

 We assume that a dynamical system, e.g., describing the ego vehicle, is governed by the dynamics 
 \begin{equation}\label{eq:DynamicalSystem}
        \dot{\mathbf{x}} = f(\mathbf{x}) + g(\mathbf{x}) u, 
    \end{equation}
 where $f, g$ are locally Lipschitz-continuous functions of the system states $\mathbf{x} \in D \subseteq \mathbb{R}^n$, $u \in U \subseteq \mathbb{R}^m$ is the input vector, and $U$ is the set of allowable inputs. 
 CBFs are used to provide safety guarantees for such systems.  
 \begin{definition}[Control Barrier Function~\cite{ames2019control}] 
 \label{def:ControlBarrierFunctions}
 Let $h(\mathbf{x}): \mathbb{R}^n \rightarrow \mathbb{R}$  be a continuously differentiable function and let $C \subseteq D \subseteq \mathbb{R}^n$ be a compact superlevel set of $h(\mathbf{x})$ such that $C = \{\mathbf{x}\in D: h(\mathbf{x}) \geq 0\}$. We say that $h$ is a \emph{Control Barrier Function} (CBF) if there exists an extended $\mathcal{K}_\infty$ class function $\alpha$ such that, for all $\mathbf{x} \in D$, the following holds:
\begin{equation}\label{eq:controlBarrierDissipativity}
    \sup_{u \in U} \: [\:L_f h(\mathbf{x}) + L_g h(\mathbf{x}) u\:] \: \ge -\alpha(h(\mathbf{x})).
\end{equation}
$L_fh(\mathbf{x}) = \frac{d h(\mathbf{x})}{d t} f(\mathbf{x})$ and $L_gh(\mathbf{x}) = \frac{d h(\mathbf{x})}{d t} g(\mathbf{x})$ are the appropriate Lie derivatives. $C$ is the safe set corresponding to the CBF~\cite{ames2019control}. 
% $\alpha$ is a $\mathcal{K}_\infty$ class function
\end{definition}
  It can be proven~\cite{ames2016control,ames2019control} that any controller $u\in \mathcal{U}_{safe}=\{u\in U: L_f h(\mathbf{x}) + L_gh(\mathbf{x})u \geq -\alpha (h(\mathbf{x}))\}$ ensures that, if the system starts in $C$, i.e., $\mathbf{x}(t_0) \in C$, then it will stay in $C$. 
  The existence of a CBF is then equivalent to ensuring the forward-invariance property of the safe set $C$, hence rendering the system evolution safe, given safety conditions on its initial states. 
  A notion of finite-time convergence CBF (FCBF) has also been proposed for time-invariant CBFs~\cite{li2018formally} 
  to guarantee finite-time convergence to a safe set. In this paper, we extend the concept of FCBF to time-varying CBF 
  % of relative degree $1$ 
  and formalize them as A/G contracts.

\subsection{Signal Temporal Logic}\label{sec:STLContracts}
We represent the mission specification using STL~\cite{maler2004monitoring}, which offers a rigorous formalism for the specification and analysis of temporal properties of real-valued dense-time signals. %
 We assume that an STL atomic predicate $\phi_h$ is evaluated over a real-valued \emph{predicate function} $h(\mathbf{x}):\mathbb{R}^n \to \mathbb{R}$. $\phi_h$ evaluates to \emph{true} ($\top$) if $h(\mathbf{x}) \ge 0$ holds, and \emph{false} ($\bot$) otherwise. 
 We then consider a fragment of STL according to the following syntax:
  \begin{equation}
        \label{eq:STLFragment}
        \psi := \top \:|\: \phi_h \:|\: \neg \phi_h \:|\: \mathcal{G}_{\Gamma}\phi_h \:|\: \mathcal{F}_{\Gamma} \phi_h \:|\: \neg \psi \:|\: \psi_1 \wedge \psi_2 \:%|\: \psi_1 \vee \psi_2  
    \end{equation}
    where $\phi_h$ is an atomic predicate, $\psi, \psi_1, \psi_2$ are STL formulas.
$\mathcal{G}$ and $\mathcal{F}$ %$\mathcal{U}_{[a, b]}$ 
are the \emph{globally} and  \emph{eventually} 
temporal operators, respectively, and $\Gamma$
is a bounded time interval. Our fragment does not include the nesting of temporal operators. 
We say that $(\mathbf{x}, t)$ satisfies $\psi$, written $(\mathbf{x}, t) \models \psi$, if there exists a signal (trajectory) $\mathbf{x}(t) \in \mathbb{R}^n$ such that $\psi$ holds at time $t$. We simply write $\mathbf{x} \models \psi$ if $(\mathbf{x}, 0) \models \psi$.
%
% \section{Sequential Composition of Piecewise-Continuous Time-Varying CBFs for STL Specifications}\label{sec:STLtoCBFs}

\section{Control Barrier Function A/G Contracts}\label{sec:STLtoCBFs}

We begin by extending classical results from finite-time convergence CBFs~\cite{li2018formally} to a new class of CBFs, which we call time-varying finite-time convergence CBFs (FCBFs). We show that FCBFs can be formalized as A/G contracts. Their composition leads, in general, to piecewise continuously differentiable ($\mathcal{C}^1$) TV-CBF contracts for which a controller is guaranteed to exist.

\begin{definition}[Time-Varying Finite-Time Convergence CBF (TV-FCBF)] \label{def: TV-FCBF-1} 

Let $h(t,\mathbf{x}): \mathbb{R}_{\geq0} \times  \mathbb{R}^n \to \mathbb{R}$ be a continuously differentiable function. Let $C(t) \subseteq D \subseteq \mathbb{R}^n$ be the compact superlevel set of $h(t,\mathbf{x})$. 
If for the system in~(\ref{eq:DynamicalSystem}) there exist $0 \leq\rho <1$ and $\gamma>0$ such that, $\forall \ t\geq 0, \forall \ \mathbf{x}\in D$, we have 
\begin{equation*} 
\sup\limits_{u\in U}\left [ \begin{split}
\frac{\partial}{\partial t}h(t,\mathbf{x})+ L_{f}h(t,\mathbf{x})+L_{g}h(t,\mathbf{x})u+ \\
+ \gamma \mathrm{sign} (h(t,\mathbf{x})) \vert h(t,\mathbf{x})\vert ^{\rho}
\end{split} \right ] \geq 0,\end{equation*}
then $h(t,\mathbf{x})$ is a Time-Varying Finite-Time Convergence CBF.
\end{definition}

When $h(t,\mathbf{x})=h(\mathbf{x})$ the TV-FCBF reduces to an FCBF~\cite{li2018formally}.  For a given $h(t,\mathbf{x})$, the set of safe inputs is 
\begin{align*} 
\label{eq: K_TV-FCBF-1}
\mathcal{U}_{safe}(t) &= \{u\in U: \frac{\partial}{\partial t}h(t,\mathbf{x}) + L_{f}h(t,\mathbf{x})+ L_{g}h(t,\mathbf{x})u +\\ &+ \gamma \mathrm{sign} (h(t,\mathbf{x})) \vert h(t,\mathbf{x})\vert ^{\rho}\geq 0\}.
\end{align*}
The following theorem discusses the finite-time convergence property of a TV-FCBF.  

\begin{theorem}\label{thm:TVCBF-timebound}

Let $C(t)$ be the superlevel set of the continuously differentiable function $h(t,\mathbf{x}): \mathbb{R}_{\geq 0}\times \mathbb{R}^n \rightarrow \mathbb{R}$ as in Definition~\ref{def: TV-FCBF-1}, with corresponding $0 \leq\rho< 1$ and $\gamma>0$, and let $\mathbf{x}(t_0)=\mathbf{x}_0 \in C(t_0)$ be the initial state, any controller $u \in \mathcal{U}_{safe}(t)$ renders  $C(t)$ forward-invariant $\forall t>t_0$. Moreover, if $\mathbf{x}_0 \in D \setminus C(t_0)$, then $u \in \mathcal{U}_{safe}(t)$ drives $\mathbf{x}$ to $C(t)$ within a finite time $T=\frac{1}{\gamma (1-\rho)} \lvert h(t_0,\mathbf{x_0}) \rvert ^{1-\rho}$. 
\end{theorem}

\begin{proof}
The proof follows the same line of reasoning as previous results~\cite{li2018formally}. Consider the candidate Lyapunov function $V(t,\mathbf{x}) = \max \left(0,- h(t,\mathbf{x})\right)$. When $\mathbf{x}_0 \in C(t_0)$, then $h(t_0,\mathbf{x})\geq 0$ and $V(t,\mathbf{x})=0$ hold. 
By the comparison lemma~\cite{bhat2000finite}, we obtain $V(t,\mathbf{x})=0$ for all $t>t_0$, hence  $\mathbf{x}(t) \in C(t)$ $\forall t>t_0$. 
When $\mathbf{x}_0\in D \setminus C(t_0)$, then we have $V(t_0,\mathbf{x}_0)>0$ and $\dot{V}(t,\mathbf{x}) \leq -\gamma V^{\rho} (t,\mathbf{x})$ $\forall t \geq t_0$. 
Again, by the comparison lemma, $\mathbf{x}$ will converge to $C(t)$ within the finite time $T=\frac{1}{\gamma (1-\rho)} \lvert h(t_0,\mathbf{x_0}) \rvert ^{1-\rho}$, i.e., $\mathbf{x}(t) \in C(t)$ $\forall t\geq t_0+T$. 
% < \infty$.  
\end{proof}

We also say that $C_h(t) = \{\mathbf{x}\in D: h(t,\mathbf{x})\geq 0\}$ is the safe set for the barrier function $h(t, \mathbf{x})$. 
TV-FCBFs can mitigate the conservatism of concatenating multiple CBFs via a $\min$ operator and taking a smooth approximation of the result, as discussed in the following example. 

\begin{example}

For a simple vehicle model, $\dot{\mathbf{x}}=\begin{bmatrix}0 & 1 \\0  & 0  \end{bmatrix} \mathbf{x}+ \begin{bmatrix}0 \\1 \end{bmatrix}u$, where $u\in \mathbb{R}$ is the input acceleration and $\mathbf{x} = \begin{bmatrix}X & V  \end{bmatrix}^T$, with $X \in\mathbb{R}$ and $V \in \mathbb{R}_{\geq 0} $ being the position and the velocity of the vehicle, let us consider the STL formula  $\bigwedge_{i \geq 1}^{N} \mathcal{G}_{[t_{i-1},t_i)}\phi_i$, 
where the predicates $\phi_i := V_{max,i}-V \geq 0$ are defined over $N$ non-overlapping adjacent intervals, that is,  $\Gamma_i=[t_{i-1},t_i)$ and $\Gamma_i \cap \Gamma_{i+1} =\emptyset$ for $1 \leq i \leq N$. 
The formula $\phi_i$ prescribes a maximum speed limit during $\Gamma_i$. 

A method to synthesize a controller satisfying this formula could 
{associate a CBF $h_i(\mathbf{x}):=V_{max,i}-V$ with each predicate $\phi_i$ and combine them via a smooth under-approximation of the pointwise $\min$ operator~\cite{lindemann2018control}}. 
However, this procedure would result into an overly conservative requirement, always yielding a speed limit that is stricter than $\underset{1 \leq i \leq N}{\min}(V_{max,i}-V)$. 
% . 
We show that a combination of TV-CBF and TV-FCBF formalized as contracts can mitigate this conservatism.
\end{example}

The notion of invariance can be formalized as a \emph{contract} to be satisfied by the closed-loop system. 
The contract algebra can then be used to reason about the composition of invariance properties and derive conditions for control synthesis.

\begin{definition}[TV-CBF Contract]\label{def:CBFContract}
A TV-CBF contract $\mathscr{C}_h = (V, A_h, G_h)$ is defined as follows:
\begin{equation}\label{eq:CBFContract}
    \begin{split}
      % \mathrm{C}_{i} := \\
     & {V} = \{ t, 
      \mathbf{x} = (x_1, \ldots, x_n)\},\\
     & {A_h} :=  \mathbf{x}(t_{0}) \in C_h(t_{0}),\\
     & {G_h} :=\forall \ t > t_0: \mathbf{x}(t) \in C_h(t),
    \end{split}
\end{equation}
where $C_h(t)$ is the safe set of $h(t, \mathbf{x})$ and $t_0$ is the initial time.
\end{definition}

\begin{definition}[TV-FCBF Contract]\label{def:TV-FCBF-Contract}
A TV-FCBF contract $\mathscr{C}_h = (V, A_h, G_h)$ is defined as follows: 
\begin{equation*}\label{eq:CBFContract}
    \begin{aligned}
    & A_h :=  \exists \ t_{conv} > 0 : \forall \ t \in [t_0, t_0 + t_{conv}): \mathbf{x}(t) \notin C_h(t) \: \wedge \\
    & \:\:\: \mathbf{x}(t_0 + t_{conv})\in C_h(t_0 + t_{conv}),\\
      & G_h :=\forall \ t \geq t_{conv} + t_0 : \mathbf{x}(t) \in C_h(t).
    \end{aligned}
\end{equation*}
\end{definition}

If a CBF contract holds, then there exists a controller which ensures the forward-invariance of $C_h(t)$ for the closed-loop system. Else, if a FCBF contract holds, there exists a controller which brings the system to the safe set $C_h(t)$ after time $t_{conv}$. 
We say that contract $\mathscr{C}_h$ 
is imposed on the interval $\Gamma = [t_0, t_f) $, when $t_0$ is the initial time and $C_h(t)$ is the safe set for all $t \in \Gamma$. Let $C_h(t_1^-)=\{\mathbf{x}:\lim_{t\to t_1^-}h(t,\mathbf{x})\geq 0\}$. In the following, we also assume that contracts are non-vacuous, that is, their assumptions are satisfiable. 
Using the above notions, we can combine contracts over neighboring intervals as stated by the following results.%, whose proofs can be found in an extended version of this paper~\cite{missing-arxiv-reference}.

\begin{theorem}\label{thm:CompositionofCBFs}
The composition of the TV-CBF contracts $\mathscr{C}_{h} \otimes \mathscr{C}_g$ imposed on $\Gamma_1 = [t_0, t_1)$ and $\Gamma_2 = [t_1, t_2)$, respectively, with $t_2>t_1>t_0$, is compatible 
if and only if $C_h(t_1^-) \cap C_g(t_1) \neq \emptyset$. Moreover,  $\mathscr{C}_h \otimes \mathscr{C}_g$  is compatible for all  $\mathbf{x}(t_0)\in C_h(t_0)$  if and only if $C_h(t_1^-) \subseteq C_g(t_1)$ holds.  
\end{theorem}

Intuitively, contracts $\mathscr{C}_h$ and $\mathscr{C}_g$ imposed on adjacent intervals can be composed, that is, \emph{at least one} trajectory $\mathbf{x}(t)$ satisfying the guarantees of $\mathscr{C}_h$ will also satisfy the assumptions of $\mathscr{C}_g$,  
if and only if their safe sets overlap at the switching point of their intervals. 
On the other hand, for \emph{every} trajectory $\mathbf{x}(t)$ 
% (with $\mathbf{x}(t_0) \in C_h(t_0)$) 
to remain safe throughout $\Gamma_1 \cup \Gamma_2$ the current safe set must be a subset of the upcoming safe set at the time of switching. This latter, stronger condition is equivalent to the condition for sequential composition of funnels~\cite{majumdar2017funnel} in Figure~\ref{fig:funnel_analogy}(a). We now present the corresponding result for FCBF contracts. 
 
\begin{theorem}\label{thm:TV-FCBF-Compatibility}
The composition of TV-CBF contract $\mathscr{C}_h$ imposed 
on $\Gamma_1 = [t_0, t_1)$ and TV-FCBF contract $\mathscr{C}_g$ imposed on $\Gamma_2 = [\tau_g, t_2)$ is compatible if and only if (1) $ C_h(t_1^-) \cap C_g(t_1) \neq \emptyset$ and (2) $ \tau_g + t_{conv, g} < t_1$ hold.
\end{theorem}

In other words, we can guarantee that every trajectory $\mathbf{x}(t)$ 
% (with $\mathbf{x}(t_0) \in C_h(t_0)$) 
remains in the safe sets of $\mathscr{C}_h$ and $\mathscr{C}_g$ throughout $\Gamma_1 \cup \Gamma_2$ by simply requiring that $C_h(t)$ overlaps with $C_g(t)$ at the time of switching, provided that $\mathbf{x}(t)$ converges to $C_g(t)$ before $\Gamma_1$ has finished. Because $t_{conv, g} \le T_g$ by Theorem~\ref{thm:TVCBF-timebound}, if $\tau_g + T_g \le t_1$, then  condition (2) in Theorem~\ref{thm:TV-FCBF-Compatibility} will also hold. % is satisfied \emph{a-fortiori} 
We leverage this observation   
combined with Definition~\ref{def: TV-FCBF-1} and the results above to compute the set of safe inputs for a composition of contracts imposed on neighboring intervals.

\begin{theorem}\label{them: piecewise TV-FCBF-1}
Let $ \mathscr{C}_{{h}} = \bigotimes_{i=1}^{\mathcal{N}}\mathscr{C}_{h, i}$ be
imposed on non-overlapping intervals $\{\Gamma_1, \Gamma_2, \dots, \Gamma_{\mathcal{N}}\}$ such that $\forall i:  1\leq i \leq \mathcal{N}$, $\Gamma_i = [t_{i-1},t_i)$ and $t_{i-1}<t_i$. Let $\mathscr{C}_{h,i}$ be either a TV-CBF contract imposed over $\Gamma_i$ or $\mathscr{C}_{h,i}=\mathscr{C}_{h,i1}\otimes \mathscr{C}_{h,i2}$ where $\mathscr{C}_{h,i1}$ is a TV-CBF contract imposed over $\Gamma_i$ and $\mathscr{C}_{h,i2}=\mathscr{C}_{h,i+1}$ is a TV-FCBF contract imposed over $[\tau_i, t_i)$ for $t_{i-1} < \tau_i < t_i $. If
$\mathbf{x}(t_0)\in C_{h_1}(t_0)$, then any controller $u \in \mathcal{U}_{safe, h}(t)$ will make the closed-loop system 
satisfy $\mathscr{C}_{h}$, where

\begin{align*}
\begin{split}
\mathcal{U}_{safe, h}(t) =& \begin{cases} \mathcal{U}_{safe, i}(t) & (t_{i-1}<t \leq \tau_i) \lor  \\
&C_{h_i}(t_i^-) \subseteq C_{h_{i+1}}(t_i)\\ 
\\
\mathcal{U}_{safe, i}(t)\: \cap & (\tau_i< t<t_{i}) \land\\\mathcal{U}_{safe, {i+1}}(t) & C_{h_i}(t_i^-) \nsubseteq C_{h_{i+1}}(t_i) \land\\
& C_{h_i}(t_i^-) \cap C_{h_{i+1}}(t_i) \neq \emptyset
\end{cases} 
\end{split}
\end{align*}
\begin{align*}
\begin{split}
\mathcal{U}_{safe, i}(t)=& \{u\in U: \frac{\partial}{\partial t}h_{i}(t,\mathbf{x}) + L_{f}h_{i}(t,\mathbf{x})+\\ & + L_{g}h_{i}(t,\mathbf{x})u + \alpha (h_{i}(t,\mathbf{x})) \geq 0\}
\end{split}
\end{align*}
\begin{align*}
\mathcal{U}_{safe, {i+1}}(t)=&\{u\in U: \frac{\partial}{\partial t}h_{i+1}(t,\mathbf{x}) + L_{f}h_{i+1}(t,\mathbf{x})+\\ & + L_{g}h_{i+1}(t,\mathbf{x})u+\gamma_{i+1} \mathrm{sign} (h_{i+1}(t,\mathbf{x}))\cdot \\ & \cdot\vert h_{i+1}(t,\mathbf{x})\vert ^{\rho_{i+1}}\geq 0\}, 
\end{align*}provided that
\begin{equation*}
\label{eq: T for piecewise TV-FCBF-1}
T_{{i+1}}=\frac{\lvert h_{i+1}(\tau_i,\mathbf{x}(\tau_i)) \rvert ^{1-\rho_{i+1}}}{\gamma_{i+1} \cdot (1-\rho_{i+1})}  < t_{i}-\tau_i \leq t_{i} -t_{i-1}.
\end{equation*}
\end{theorem}

For a given set of STL tasks as in~\eqref{eq:STLFragment}, we propose a control synthesis algorithm by mapping the tasks to contracts imposed on neighboring intervals. 

\section{Synthesis Algorithm}

We propose a two-step control synthesis algorithm. In the pre-processing step, the high-level planner provides a set of STL tasks $\{\psi_1 =\mathcal{T}_{\Gamma_1}(h_1(\mathbf{x}) \ge 0), \ldots, \psi_\mathcal{N} = \mathcal{T}_{\Gamma_\mathcal{N}}(h_\mathcal{N}(\mathbf{x}) \ge 0)\}$, with $\mathcal{T} \in \{\mathcal{G},\mathcal{F}\}$, defined over the time intervals $\{\Gamma_1, \ldots, \Gamma_\mathcal{N}\}$ such that, $\forall i: 1 \leq i \leq \mathcal{N}$, $t_{i-1} < t_i$, and $\Gamma_i = [t_{i-1},t_i)$. %
Predicates of the form $\mathcal{F}_{\Gamma_i} (h_i(\mathbf{x}) \ge 0)$ are converted to $\mathcal{G}_{[t_s, t_s+\epsilon)} h_i(\mathbf{x})\ge 0$, %
where $[t_s, t_s+\epsilon) \subset \Gamma_i$ is a user-specified time of satisfaction and  $\epsilon>0$. %
We then divide the set of STL tasks into the minimum number of STL formulas $\mathtt{G}_1, \ldots, \mathtt{G}_k$ such that each $\mathtt{G}_j$ is a conjunction of predicates %$(1)$ 
$\mathcal{G}_{\Gamma_i}\phi_{h_i}$ defined over non-overlapping intervals, as in  Theorem~\ref{them: piecewise TV-FCBF-1}. For example, given the specification $\Psi = \mathcal{G}_{[0, T_{max})}(h(\mathbf{x}) \geq 0) \wedge \bigwedge_{i=1}^{N} \mathcal{G}_{[t_{i-1},t_i)} (h_{i}(\mathbf{x}) \geq 0$), we can separate its predicates into $\mathtt{G}_1 = \mathcal{G}_{[0, T_{max})}(h(\mathbf{x}) \geq 0)$ and  $\mathtt{G}_2 = \bigwedge_{i = 1}^{N}\mathcal{G}_{[t_{i-1},t_i)} (h_{i}(\mathbf{x}) \geq 0)$ such that $\Psi = \mathtt{G}_1 \wedge \mathtt{G_2}$.

Next, we proceed to the synthesis step. Each $\mathcal{G}_{\Gamma_i}\phi_{h_i}$ in  $\mathtt{G}_j$ corresponds to the predicate function $h_i(\mathbf{x})$. Because 
% Note that 
any $\mathbf{x}(t)$ in the safe set $C_{h_i}(t)$
satisfies $\mathcal{G}_{\Gamma_i}\phi_{h_i}$,  % Therefore, 
control synthesis for $\mathtt{G}_j$ reduces to finding a controller for a contract of the form $\mathscr{C}_j = \bigotimes_{i = 1}^{|\mathtt{G}_j|} \mathscr{C}_{h, i}$. Consequently, we use Theorem~\ref{them: piecewise TV-FCBF-1} to construct the safe input set $\mathcal{U}_{safe, j}(t)$ by conditioning over two possibilities: (1) $C_{h_i}(t_i^-) \subseteq C_{h_{i+1}}(t_i)$ or (2) $C_{h_i}(t_i^-) \cap C_{h_{i+1}}(t_i) \neq \emptyset$ and $T_{{i+1}} \le t_i - t_{i-1}$. 
 We can synthesize a controller to simultaneously satisfy $\mathscr{C}_{ 1}, \ldots, \mathscr{C}_{k}$, corresponding to $\mathtt{G}_1, \ldots, \mathtt{G}_k$, by resorting to contract conjunction (see Section~\ref{subsec:ContractBasedDesign}). 
 Satisfying $\bigwedge_{j = 1}^k \mathscr{C}_{j}$ corresponds to finding a controller in the safe set $\mathcal{U}_{safe, f}(t) = \bigcap_{j = 1}^k \mathcal{U}_{safe, j}(t)$.
 Finally, the safe control actions at each time step can be obtained by solving a quadratic program of the form $u_{safe}= \underset{u \in \mathcal{U}_{safe}(t)}{\mbox{argmin}}\; ||u-u_{nom}||_2^2$ , where $u_{nom}$ is a nominal PID or any other feedback controller. The procedure is summarized in Algorithm~\ref{algo:ControlDesignAlgo}.

\begin{algorithm}[t]
\small
\KwData{ $\Psi = \bigwedge_{i=1}^I \mathcal{G}_{\Gamma_i} (h_i(\mathbf{x}) \ge 0) \wedge \bigwedge_{j=1}^J \mathcal{F}_{\Gamma_j} (h_j(\mathbf{x}) \ge 0) $, times of satisfaction $\mathbf{t} = {t_{s, 1}}, \ldots, t_{s, J}$}
 \KwResult{$u_{safe}$ or \texttt{failure}}
 $\Psi' = \bigwedge_{i=1}^{I} \mathcal{G}_{\Gamma_i} (h_{i}(\mathbf{x}) \ge 0) \wedge \bigwedge_{j=1}^{J} \mathcal{G}_{[t_{s, j}, t_{s, j} +\epsilon )} (h_{j}(\mathbf{x}) \ge 0)$;\\
 
 /*Divide STL into formulas $\mathtt{G}_j$ whose predicates are defined on non-overlapping intervals $\Gamma_{1},\dots, \Gamma_{\mathcal{N}_j}$, $\Gamma_i=[t_{i-1},t_i)$*/  \\
 $\mathtt{G}_1,\mathtt{G_2},\dots,\mathtt{G}_k$ $=$ $\mathtt{group}(\Psi')$;\\

  \For{$\mathtt{G}_j \in \{\mathtt{G}_1,\mathtt{G}_2,\dots,\mathtt{G}_k\}$}{
  \For{$h_i \in \mathtt{G_j}$}{
  $comp = \mathtt{chk-compatibility}$($\mathscr{C}_{h,i} \otimes \mathscr{C}_{h,i+1}$)\\
  /* Check (1) $C_{h_i}(t_{i}^-) \subseteq C_{h_{i+1}}(t_{i})$ or (2) $C_{h_i}(t_{i}^-) \cap C_{h_{i+1}}(t_{i}) \neq \emptyset$ and $T_{{i+1}} < t_{i}-t_{i-1}$.*/\\
  \If{$comp = \bot$}{
  \KwRet{\texttt{failure}}\\}
  }
   
 }
\For {$t \leq T_{max}$}{
%   /* Synthesize the controller $u(t)$ */ 
\begin{align}
% \begin{split}
%\mathcal{U}_{safe, f}(t)& =U\cap \bigcap_{j=1}^{k} \mathcal{U}_{safe, j}(t)
\mathcal{U}_{safe, f}(t) = \bigcap_{j=1}^{k} \mathcal{U}_{safe, j}(t)
%\mathcal{U}_{safe}(t)& =U\cap \bigcap_{j=1}^{\mathtt{k}} \mathcal{U}_{\mathfrak{hj}}(t)
% \end{split}
\end{align}
/* $\mathcal{U}_{safe, j}(t)$ is calculated as in Theorem~\ref{them: piecewise TV-FCBF-1}*/ 

\If {$\mathcal{U}_{safe, f}(t) \neq \emptyset$}{
% /* Nominal Controller $u_{nom}$: PID or any other control law */
%\KwRet{$\mathbf{u}$};
\begin{align}
\label{eq:u_safe_using_QP}
\begin{split}
u_{safe} &= \underset{u \in \mathcal{U}_{safe, f}(t)}{\mbox{argmin}}\; ||u-u_{nom}||_2^2\\
% \dot{\mathbf{x}} &= f(\mathbf{x}) + g(\mathbf{x})u_{safe}
\end{split}
\end{align}
} 
\Else{
\KwRet{\texttt{failure}}
}
}
\caption{Control Synthesis Algorithm} \label{algo:ControlDesignAlgo}
\end{algorithm}

\section{Case Study}\label{sec:casestudy}

We illustrate our algorithm on a vehicular planning problem under safety and regulatory constraints. 
We use a non-linear model for the ego vehicle. Let $X_f$, $V_f$, $X_l$, and $V_l$ be the longitudinal position and velocity of the ego vehicle and the lead vehicle, respectively. Let $X_l>X_f$ and let $X_r=X_l-X_f>0$ be the relative distance and $V_r=V_l-V_f$ the relative velocity between the ego and the lead vehicle. Let $F_r=c_0+c_1 V_f + c_2 V_f^2$ be the sum of all frictional and aerodynamic forces on the ego vehicle, where $c_0$, $c_1$, and $c_2$ are parameters.  Let $u \in U $ be the input wheel force, $\mathbf{x}=\begin{pmatrix} X_f & V_f & X_l \end{pmatrix}^T$, $f(\mathbf{x}) = \begin{pmatrix} V_f & -\frac{1}{m}F_r & V_l \end{pmatrix}^T$, and $g(\mathbf{x}) = \begin{pmatrix} 0 & \frac{1}{m}& 0 \end{pmatrix}^T$. The longitudinal dynamics of the system are given by $\dot{\mathbf{x}} = f(\mathbf{x}) + g(\mathbf{x})u$ while $t_{\mathcal{N}_v}$ is the time horizon of the mission.

\subsection{STL Specifications.}  
The mission includes safety and regulatory requirements as follows:

\begin{enumerate}

\item The ego vehicle must maintain a safe distance from the lead vehicle. Let $h_1(\mathbf{x})$ % := \phi_{h_1}$ 
be the spacing error between the ego vehicle and the lead vehicle, defined by $h_1 := X_r-hV_f-S_0 -\frac{V_f^2-V_l^2}{2a_{max}}$, 
% \textcolor{blue}{$\phi_1 = X_r-hV_f-S_0 -\frac{V_f^2-V_l^2}{2|a_{max}|}$}. 
where $h$ is a constant time-headway (time required by the ego vehicle to cover the distance $X_r$), $S_0>0$ is the relative distance between ego and lead vehicle when they are in the rest position, and $a_{max}>0$ is the absolute value of the maximum braking capability. 
We then require $\mathcal{G}_{[0,t_{\mathcal{N}_v})}\phi_{h_1} 
 = \mathcal{G}_{[0,t_{\mathcal{N}_v})} (h_1(\mathbf{x}) \ge 0)$.

\item The ego vehicle must follow variable speed limits. Let the mission interval $\Gamma_v=[t_0, t_{\mathcal{N}_v})$ be divided into $\mathcal{N}_v$ consecutive intervals such that $\Gamma_v = \bigcup\limits_{1 \le i \leq \mathcal{N}_v}[t_{i-1}, t_i)$ and $t_{i-1}<t_i$ for all $i$. Let the speed limit $V_{max}(t)$ be a discrete-valued function of time, described in a piecewise manner as $V_{max}(t) = V_{max,i}$ $\forall t \in \Gamma_{vi}=[t_{i-1}, t_i)$, $1\leq i \leq \mathcal{N}_v$. We define $h_{v,i}(\mathbf{x}) := V_{max,i}(t)-V_f$ and $\phi_{h_{v,i}}:= h_{v,i}(\mathbf{x}) \geq 0$. We then require 
$\mathcal{G}_{[0,t_{\mathcal{N}_v})}\phi_v = \bigwedge\limits_{1\le i \le  \mathcal{N}_v} \mathcal{G}_{\Gamma_{vi}} \phi_{v,i}$. Let $T_{conv,v}$ be the convergence time for  $\phi_v$, with $T_{conv,v} < t_i-t_{i-1}$, for all $i$ such that $1 \leq i \leq \mathcal{N}_v$.  

\item The ego vehicle must follow the traffic signals. Let $N$ be the number of signals, and $s_i \in \{\mathtt{Red, Yellow, Green}\}$ be the state of the $i^{th}$ traffic signal. Let $g_{ij}$, $y_{ij}$, and $r_{ij}$ be the instants at which the $i^{th}$ signal turns green, yellow, and red for the $j^{th}$ time, respectively, with $g_{ij}<y_{ij}<r_{ij}<g_{i,j+1}$. We obtain that $s_i=\mathtt{Green}$ $\forall t \in \Gamma_{ij,g}=[g_{ij},y_{ij})$, $s_i=\mathtt{Yellow}$ $\forall t \in \Gamma_{ij,y}=[y_{ij},r_{ij})$, and $s_{ij}=\mathtt{Red}$ $\forall t \in \Gamma_{ij,r}=[r_{ij},g_{i,j+1})$. Let $s_i \neq \mathtt{Red}$ $\forall t\in \Gamma_{ij,\bar{r}}=[g_{ij},r_{ij})$. 
Let $h_{r,i}(\mathbf{x})= P_{i} - X_f - \beta V_f  - S_0$ and  $h_{\bar{r},i}(\mathbf{x})= P_{i+1} - X_f - \beta V_f  - S_0$, where $P_i$ is the position of the $i^{th}$ traffic signal,   $\beta>0$, and $P_{i-1}<X_f \leq P_{i} $. If $s_{i}=\mathtt{Red}$, then $\phi_{h_{r,i}} := (h_{r,i}(\mathbf{x})\geq 0)$ must hold; if $s_{i}\neq \mathtt{Red}$, then $\phi_{h_{\bar{r},i}} := (h_{\bar{r},i}(\mathbf{x}) \geq 0)$ must hold. Let $\mathbf{1}_{B}(\mathbf{x}):\mathbb{R}^n\rightarrow \{\top, \bot\}$ be an indicator function such that $\mathbf{1}_{B}(\mathbf{x})=\top$ if and only if $\mathbf{x}\in B \subseteq \mathbb{R}^n$. 
% and $\mathbf{1}_{B}(\mathbf{x})=\bot$ for $\mathbf{x}\notin B \subseteq \mathbb{R}^n$. 
Let $\mathtt{T}_i = \{\mathbf{x}:P_{i-1}<X_f\leq P_i\}$, for $1\leq i \leq N$. We can encode the traffic rules as $\bigvee_{i=1}^{N}\bigl(\bigwedge_{j\in \mathbb{N}} \bigl[ \mathcal{G}_{\Gamma_{ij,r}}\phi_{h_{r,i}}\wedge \mathbf{1}_{\mathtt{T}_i}(\mathbf{x}) \bigr] \wedge \bigl[ \mathcal{G}_{\Gamma_{ij,\bar{r}}}\phi_{h_{\bar{r},i}}\wedge \mathbf{1}_{\mathtt{T}_i}(\mathbf{x}) \bigr] \bigr)$. Therefore,  % $\forall i: 1\leq i\leq N $, 
$\exists k: 1\leq k\leq N$ such that $\mathbf{1}_{\mathtt{T}_k}(\mathbf{x})=\top$ and $\mathbf{1}_{\mathtt{T}_i}(\mathbf{x})=\bot$, $\forall i \neq k$. 
Let $h_{pos}(t,\mathbf{x})=\sum_{i=1}^{N} (h_{pos,i}(t,\mathbf{x}) \cdot \mathbf{1'}_{\mathtt{T}_i}(\mathbf{x}))$, where  $h_{pos,i}(t,\mathbf{x})=h_{r,i}(\mathbf{x})$, $\forall t\in \Gamma_{ij,r}$, $h_{pos,i}(t,\mathbf{x})=h_{\bar{r},i}(\mathbf{x})$, $\forall t\in \Gamma_{ij,\bar{r}}$, and $\mathbf{1'}_{B}:\mathbb{R}^n\to \{0,1\}$ is defined such that $\mathbf{1'}_{\mathtt{B}}(\mathbf{x})=1$ if and only if $\mathbf{x}\in \mathtt{B}$. Let the duration of the yellow signal of the $j^{th}$ cycle of the $i^{th}$ traffic signal be $Y_{ij}=r_{ij}-y_{ij}$, and let $Y_{ij}$ be the convergence time of $h_{r,i}(t,\mathbf{x})$ and $h_{\bar{r},i}(t,\mathbf{x})$. 
\item % V2V and V2I communication: 
The lead vehicle communicates its velocity $V_l$ and acceleration $a_l$  to the ego vehicle. 
\end{enumerate}

\noindent We organize the specifications above into STL formulas satisfying the conditions in Algorithm~\ref{algo:ControlDesignAlgo} and map them to compositions of contracts.  
\subsection{Mapping STL Formulas to Contracts}
%\noindent\textbf{Mapping STL Formulas to Contracts.} 
The overall STL specification before pre-processing is given by $\phi_m= \bigwedge_{i=1}^3\mathtt{G_i}$, where $\mathtt{G_1}=\mathcal{G}_{[0,t_{\mathcal{N}_v})}\phi_{h_1}$, $\mathtt{G_2} = \bigwedge_{k=1}^{\mathcal{N}_v} \mathcal{G}_{[t_{k-1},t_k)} \phi_{h_{v,i}}$, $\mathtt{G_3}=\bigvee_{i=1}^{N}\bigl(\bigwedge_{j\in \mathbb{N}} \bigl[ \mathcal{G}_{\Gamma_{ij,r}}\phi_{h_{r,i}}\wedge \mathbf{1}_{\mathtt{T}_i}(\mathbf{x}) \bigr] \wedge \bigl[ \mathcal{G}_{\Gamma_{ij,\bar{r}}}\phi_{h_{\bar{r},i}}\wedge \mathbf{1}_{\mathtt{T}_i}(\mathbf{x}) \bigr] \bigr)$.
The predicates in the specification are allocated to three  formulas such that 
the predicates in each formula are defined on non-overlapping intervals. 
%
% \subsection{Contract Compositions}
We first consider $\mathtt{G}_1$. Since  $ h_1(\mathbf{x})$ is $\mathcal{C}^1$ and there is only one STL predicate over the whole horizon, we generate the CBF contract $\mathscr{C}_{{1}}$
with the corresponding safe set $C_1$.

$\mathtt{G}_2$ can be associated to a CBF contract $\mathscr{C}_{2}=\bigotimes_{i=1}^{\mathcal{N}_v}\mathscr{C}_{h_v, i}$, where $\mathscr{C}_{h_v, i}$ can either be a CBF contract over interval $\Gamma_i$ or $\mathscr{C}_{h_v, i} = \mathscr{C}_{h_v, i1}\otimes \mathscr{C}_{h_v, i2}$, where $\mathscr{C}_{h_{v}, i1}$ is a CBF contract for the safe set $C_{h_{v, i}}$ and $\mathscr{C}_{h, i2}$ is a FCBF contract for $C_{h_{v, i+1}}$, imposed over the intervals $\Gamma_i$ and $[t_i-T_{conv,v}, t_{i})$, respectively. 
Given %$\phi_{v,i}(\mathbf{x})=  
$h_{v,i}(\mathbf{x})=V_{max,i}-V_f$, we set  
$h_v(t,\mathbf{x})=h_{v, i}(t,\mathbf{x})$ for $t\in \Gamma_i$ $\forall i \in \mathcal{N}_v$. 

Similarly to $\mathscr{C}_2$, we can form $\mathscr{C}_{3}$ for $\mathtt{G}_3$, which specifies the traffic signals' constraints for the system when the ego vehicle is approaching the $i^{th}$ traffic signal. 
When $P_{i-1} < X_f \leq P_i$, then we have $C_{pos,i}(t) =\left\{ \mathbf{x} \in \mathbb{R}^n: h_{pos,i}(t,\mathbf{x}) \geq 0 \right\}$.

At $t=r_{ij}$, we only require that $C_{pos,i}(t=r_{ij}^-) \cap C_{pos,i}(t=r_{ij}) \neq \emptyset$ 
rather than $C_{pos,i}(t=r_{ij}^-) \subseteq C_{pos,i}(t=r_{ij})$. 

\begin{figure*}
\centering
\setkeys{Gin}{width=0.22\textwidth}
\subfloat[%$\mbox{HC-I}$:
\small
$\mathcal{G}_{[0,t_{\mathcal{N}_v})}h_1(x)\geq 0$. The ego vehicle always maintains a safe distance from the lead vehicle.
          \label{fig:h1_vs_Time}]{\includegraphics{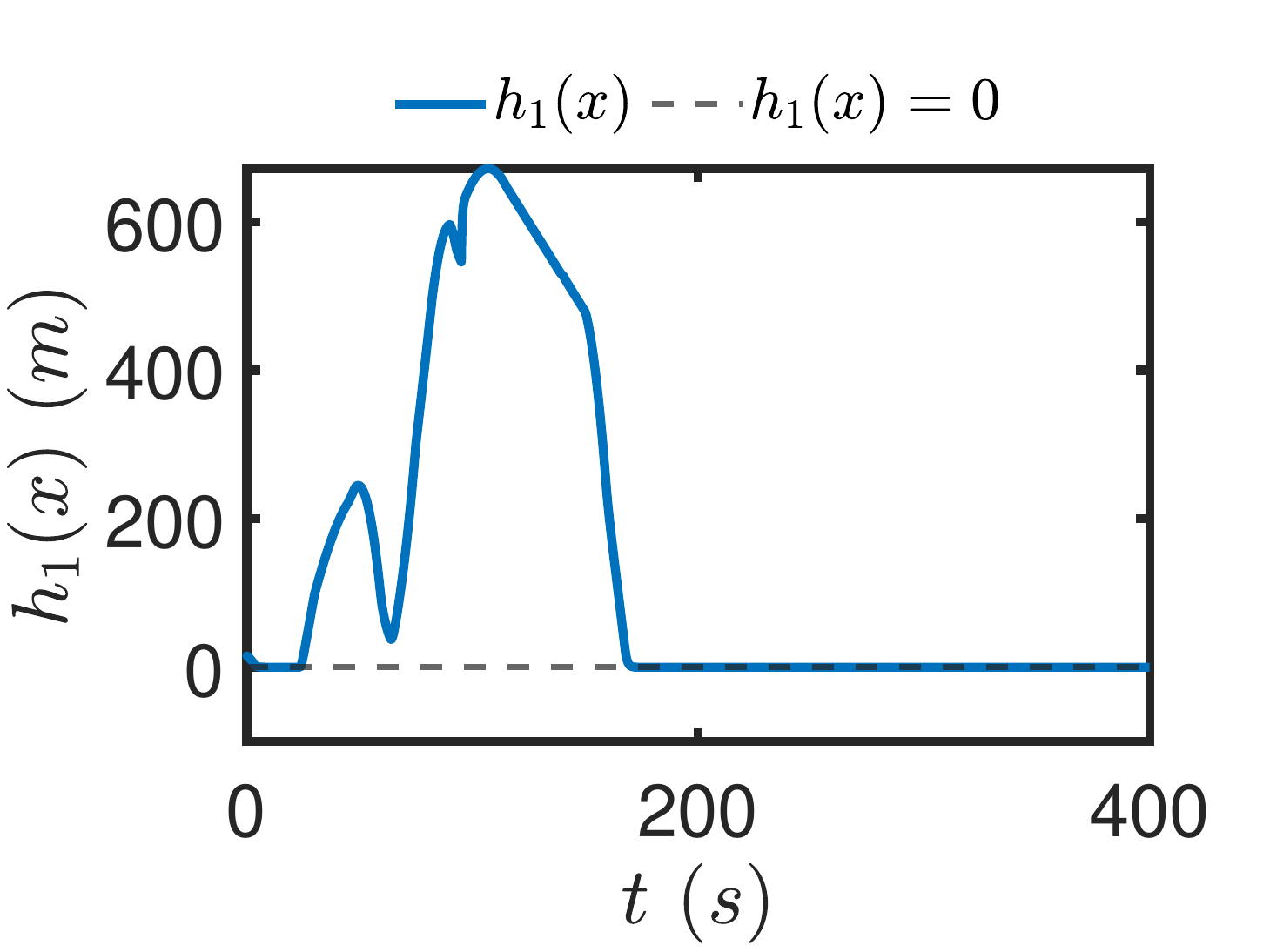}}
    \hfill
\subfloat[% $\mbox{HC-II}$:
\small
$
% \mathcal{G}_{[0,\infty)}\phi_v=
\mathcal{G}_{[0,\mathcal{N}_v)}h_{v}(t,\mathbf{x}) \geq 0$. Speed profile of the lead vehicle ($V_l$) and ego vehicle ($V_f$), and speed limit $V_{max}(t)$. 
%in $Km/h$ vs. time (s). 
         \label{fig:Vr_vs_Time}]{\includegraphics{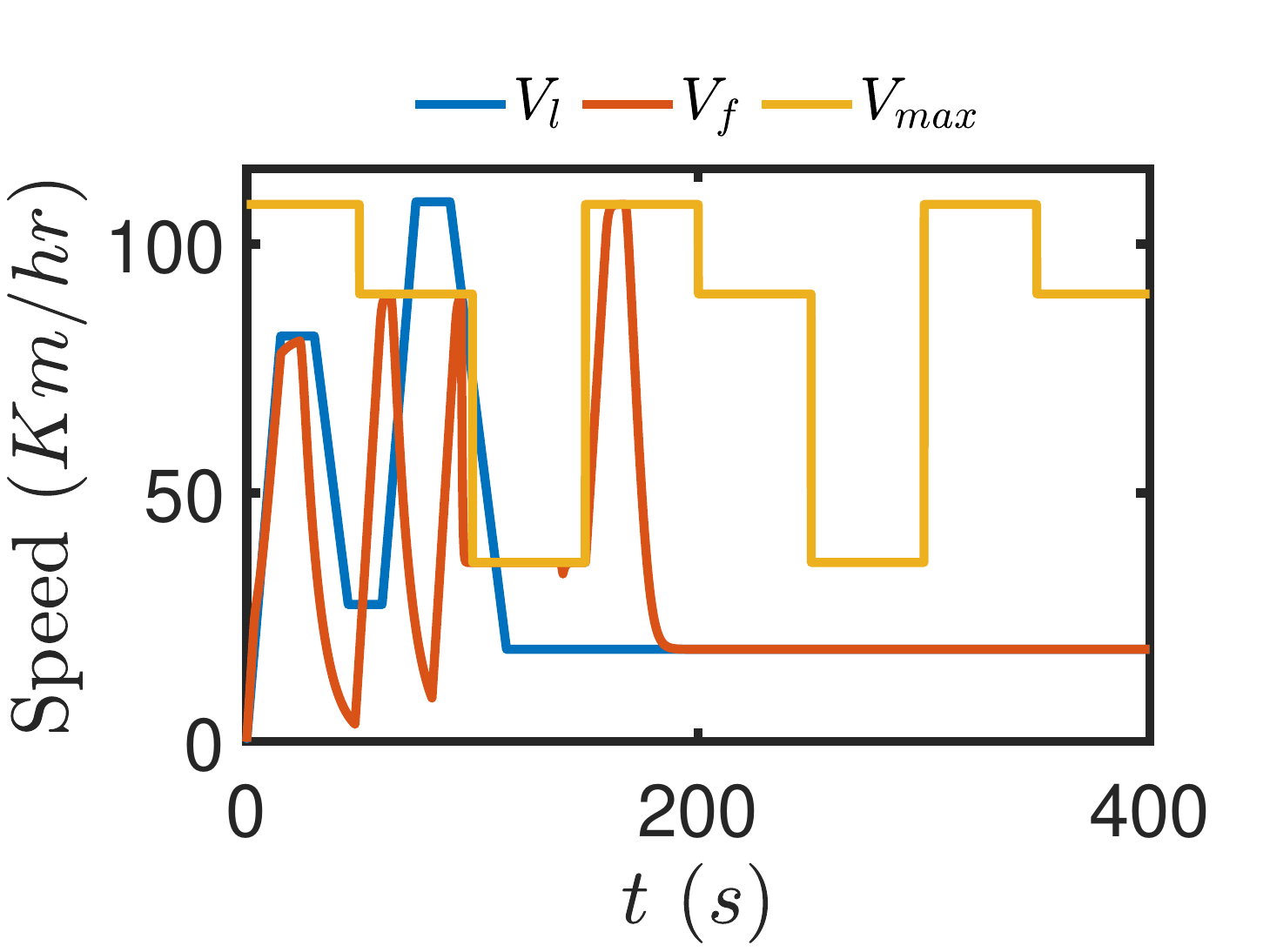}}
    \hfill
    \subfloat[% $ \mbox{HC-III:}
    % \mathcal{G}_{[0,+\infty)}\phi_r=
    \small
    $\mathcal{G}_{[0,\mathcal{N}_v)} h_{pos}(t,\mathbf{x})\geq 0$. The ego vehicle obeys the traffic signals. 
          \label{fig:hpos_vs_Time}]{\includegraphics{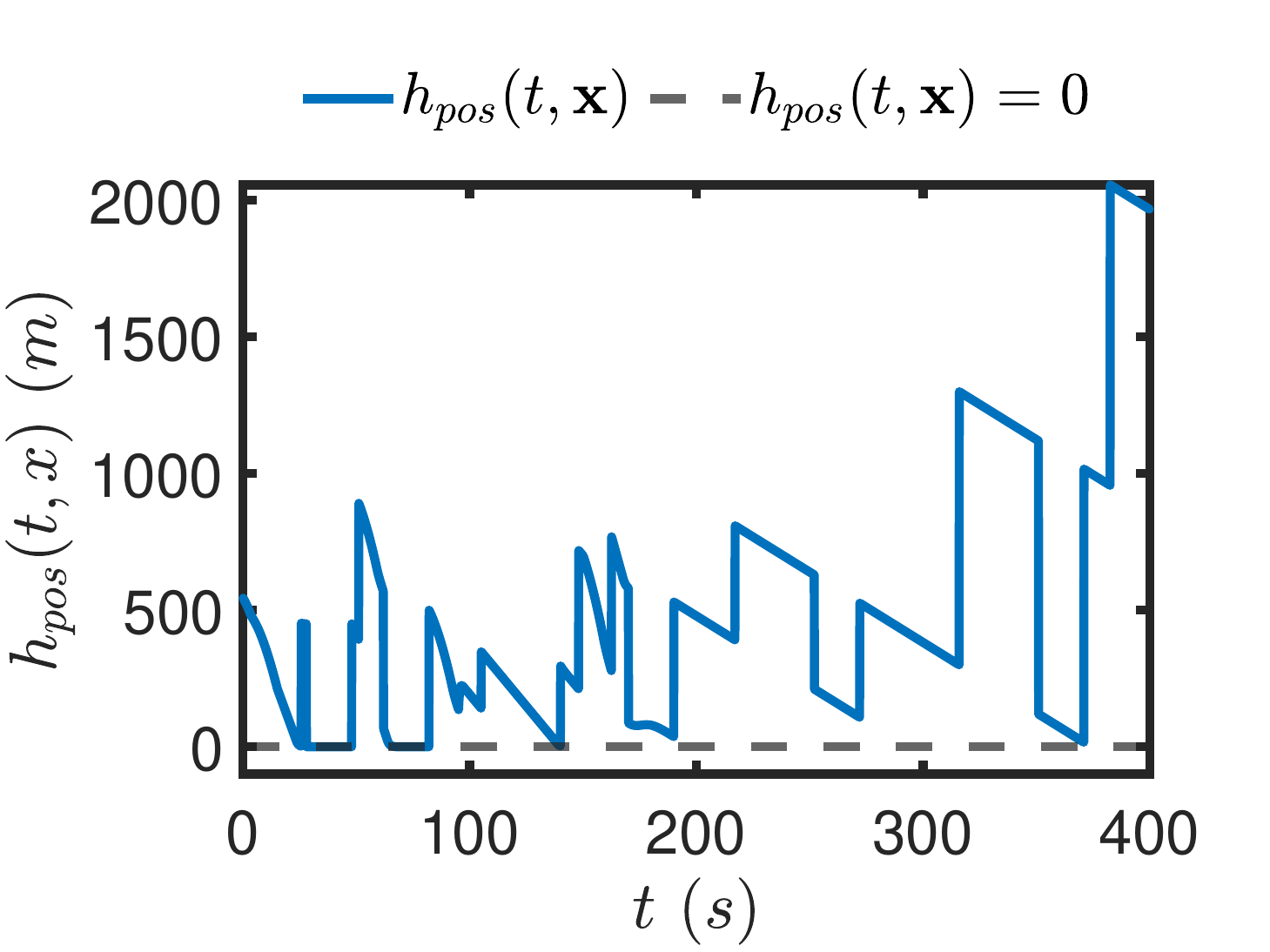}}
          \hfill
          \subfloat[\small Normalized input $u$.  
          \label{fig:U_vs_Time}]{\includegraphics{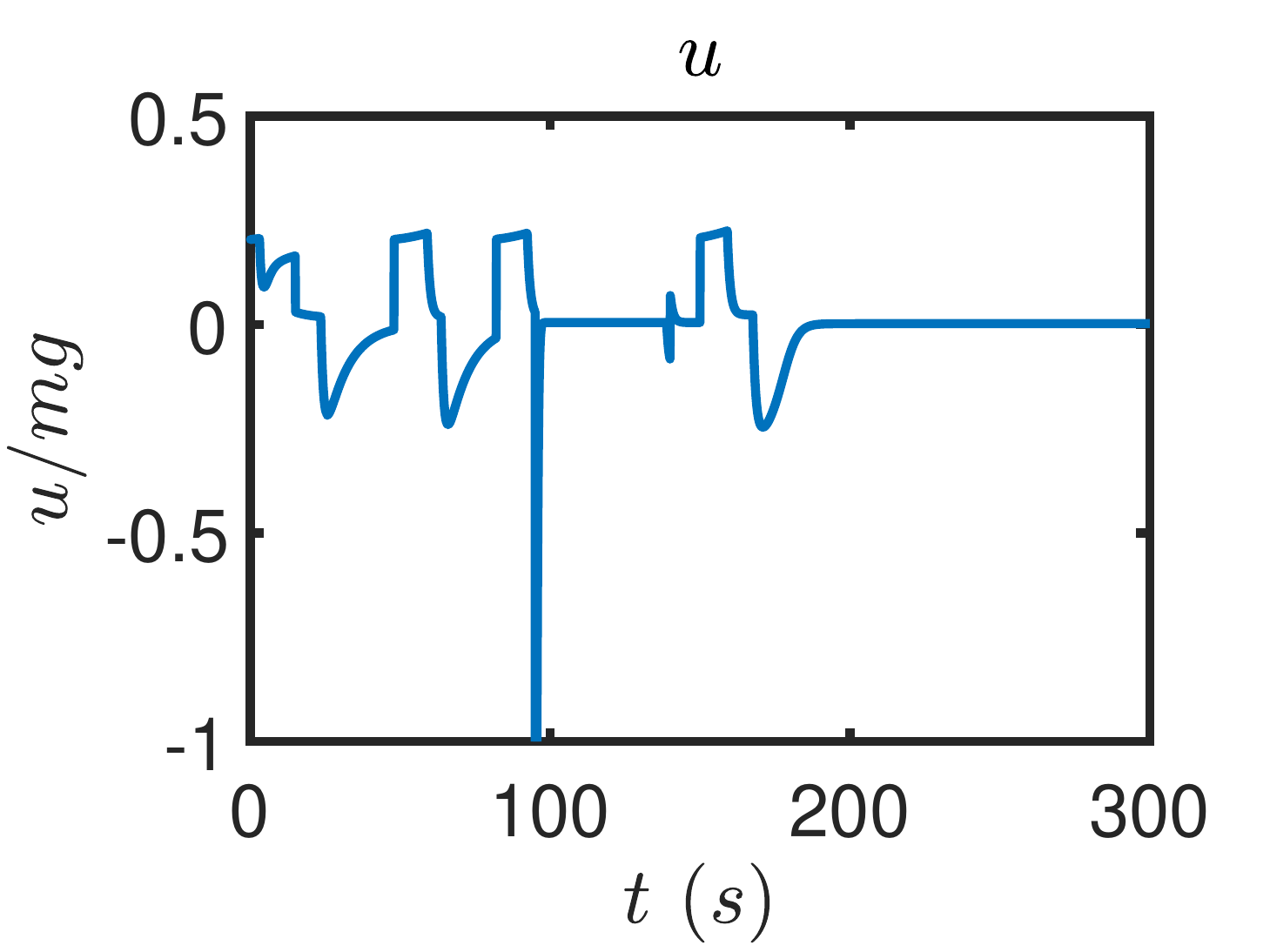}}
\caption{\small Simulation results: all safety and regulatory constraints are satisfied by the controlled vehicle.}
\label{fig:traffic_signals}
\end{figure*}

\begin{figure}[t]
        \centering \includegraphics[width=7cm]{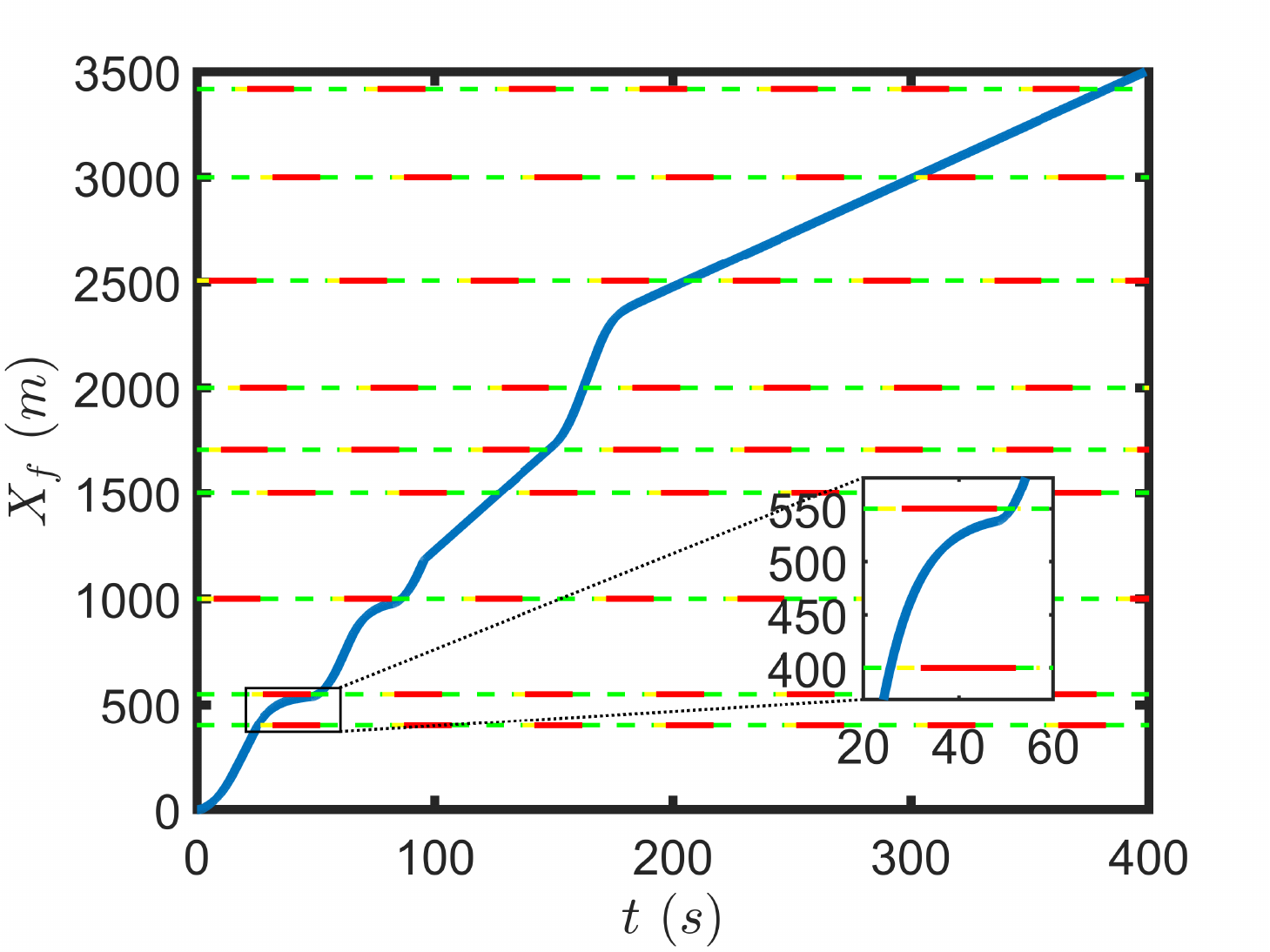}
         \caption{\small Position of the ego vehicle $X_f$ vs. time. The horizontal lines show the position of the traffic signals while the colors show the states at a given time.}
         \label{fig:Xf_vs_Time}
\end{figure}

%\noindent\textbf{Control Synthesis.}
\subsection{Control Synthesis} 
We derive the set of constraints that will be used to find the control law guaranteeing the satisfaction of contracts $\mathscr{C}_{1}$, $\mathscr{C}_{2}$, and $\mathscr{C}_{3}$.
By Definition~\ref{def:ControlBarrierFunctions} and the related treatment,
the set of inputs ensuring the 
satisfaction of $\mathscr{C}_{\mathtt{1}}$ is given by
$\mathcal{U}_{h_1}= \{u \in U: u \leq  \frac{ma_{max}}{ha_{max}+V_f} (h_1+V_r+\frac{V_l}{a_{max}}a_l)+F_r\}$, where $a_l$ is the acceleration of the lead vehicle received by V2V communication.

Let $t_{i-1,0}$ be the time at which the ego vehicle crosses the $(i-1)^{th}$ traffic signal, serving as the initial condition as the vehicle approaches the $i^{th}$ traffic signal. Let $\mathcal{U}_{g,i}(t)$ denote the set of safe inputs when $s_i = \mathtt{Green}$, i.e., when $t\in \Gamma_{ij,g} = \bigl[\max(g_{ij},t_{i-1,0}),y_{ij} \bigr)$, and $\mathcal{U}_{r,i}(t)$ the set of safe inputs  
% $h_{pos,i}(t,\mathbf{x})$.
when $s_i=\mathtt{Red}$, i.e., over $\Gamma_{ij,r}$. For all $j\in \mathbb{N}$, for all $t\in \Gamma_{ij,g}$, 
if $\max(y_{ij},t_{i-1,0})\leq t <r_{ij}$, according to Theorem \ref{them: piecewise TV-FCBF-1}, the set of safe inputs is  
\begin{align*}
\begin{split}
% \forall j& \in \mathbb{N},\; \forall t\in \Gamma_{ij,g} \\
\mathcal{U}_{h_{pos,i}}(t) =& \begin{cases} \mathcal{U}_{h_{\bar{r},i}}(t) & g_{ij}\leq t \leq y_{ij} \\
\mathcal{U}_{h_{\bar{r},i}}(t) \cap \mathcal{U}_{h_{r,i}}(t) & y_{ij}< t<r_{ij}\\
\mathcal{U}_{h_{\bar{r},i}}(t) & r_{ij}< t<g_{i,j+1}\\
\end{cases} \\
\mathcal{U}_{h_{\bar{r},i}}(t)=& \{u\in \mathcal{U}: u \leq \frac{m}{\beta}(h_{\bar{r},i}(\mathbf{x})-V_f)+F_r \}\\
\mathcal{U}_{h_{r,i}}(t)=&\bigl\{u\in \mathcal{U}:u\leq \frac{m}{\beta}\bigl(\gamma_{r,j+1}\text{sign}(h_{r,i}(\mathbf{x}))\\&\cdot\vert h_{r,i}(\mathbf{x})\vert ^{\rho_{r,j+1}}\bigr)+F_r\bigr\}. 
\end{split}
\end{align*}
We have $\gamma_{r,j+1} = \frac{\lvert h_{r,i}(t_{sj},\mathbf{x}(t_{sj})) \rvert ^{1-\rho_{j+1}}}{(Y_{ij}) \cdot (1-\rho_{j+1})}$ for $t_{sj}= \max(y_{ij},t_{i-1,0})$. 
The set inputs $\mathcal{U}_{h_v}(t)$ ensuring the satisfaction of $\mathscr{C}_{h_v}$
% for piecewise $h_v(t,\mathbf{x})$ 
can be finally calculated using Theorem \ref{them: piecewise TV-FCBF-1} $\forall j \in \mathcal{N}_v, \; \forall t\in \Gamma_{vj}$, leading to 
\begin{align*}
\begin{split}
\mathcal{U}_{h_v}(t) =& \begin{cases} \mathcal{U}_{h_v,j}(t) & t_{j-1}<t \leq t'_j \\
\mathcal{U}_{h_v,j}(t) \cap \mathcal{U}_{h_v,j+1}(t) & t'_j< t<t_{j}
\end{cases} \\
\mathcal{U}_{h_v,j}(t)=& \{u\in \mathcal{U}: u \leq \frac{m}{\beta}h_{v,j}+F_r\}\\
\mathcal{U}_{h_v,j+1}(t)=&\bigl\{u\in \mathcal{U}:u \leq \frac{m}{\beta}(\gamma_{v,j+1} \text{sign} (h_{v,j+1}) \\& \cdot\vert h_{v,j+1}\vert ^{\rho_{v,j+1}})+F_r\bigr\}. 
\end{split}
\end{align*}
We have $\gamma_{v,j+1} = \frac{\lvert h_{v,j+1}(t'_j,\mathbf{x}(t'_{j})) \rvert ^{1-\rho_{v,j+1}}}{(T_{conv,v}) \cdot (1-\rho_{v,j+1})}$, $t'_j=t_j-T_{conv,v}$%
% The $\gamma_{v,j+1}$ 
which ensures that the system converges from $C_{v,j}$ to $C_{v,j+1}$ within $T_{conv,v}$. 

By combining all the constraints above, we get $\mathcal{U}_{safe}(t)=\mathcal{U}_{safe,i}(t)$ for $P_{i-1}<X_f\leq P_i$,  $\mathcal{U}_{safe,i}(t) = \mathcal{U}_{h_1} \cap \mathcal{U}_{h_{pos,i}}(t) \cap \mathcal{U}_{h_v}(t)$. A PID controller is chosen as a nominal controller $u_{nom}=m\left(k_1V_r+k_2\phi_1 + k_3 \int_{0}^{t} \phi_1 \,dx\right)+F_r$ to achieve the goal that $V_f \rightarrow V_l$ and $\phi_1 \rightarrow 0$, where  $k_1,k_2,$ and $k_3$ are the PID gains. The $u\in \mathcal{U}_{safe}$ is obtained by solving quadratic programs as in~\eqref{eq:u_safe_using_QP}.

\section{Simulations and Results \label{sec:simulations}}

We simulate a road scenario with the ego and lead vehicles. The traffic is regulated by ten traffic signals. The distance between any two consecutive traffic signals is not equal. The timing cycles of different traffic signals are also different. The speed limits changes every $50$~s, which is possibly unrealistic but helps validate our methodology in extreme conditions. The variable speed limit $V_{max}(t)$ has three values: $30\;m/s, 
% (108\;km/hr), 
25\;m/s$, 
% \;(90\;km/hr),$ 
and $10\;m/s
% \;(36\;km/hr)
$. 
% The ego vehicle is obeying traffic lights, and variable speed limit. The lead vehicle is not following any of these regulations. 
We consider the worst-case scenario in which the lead vehicle may violate the rules.
We use $g=9.8\;\frac{m}{s^2}$, $m=1650~Kg$, $c_0=0.1\:N$, $c_1= 5\: \frac{N}{m/s}$, $c_2=0.25\frac{N}{m/s^2}$, $a_{max}=0.4\mbox{g} \frac{m}{s^2}$, and $\rho_{r,j+1}=0.9$ and $\rho_{v,j+1}=0.91$ for all $j\in \mathbb{N}$.
% For piecewise $\mathcal{C}^1$ $h_{pos}(t,\mathbf{x})$,
$Y_{ij}$ is the convergence time for $\mathscr{C}_{3}$ and $T_{conv,v}=5$~s is the convergence time within each interval $\Gamma_{vj}$ in  $\mathtt{G}_{2}$.
% For piecewise $\mathcal{C}^1$ $h_{v}(t,\mathbf{x})$, 
%
Fig.~\ref{fig:h1_vs_Time} shows that $\mathcal{G}_{[0,t_{\mathcal{N}_v})}h_1(\mathbf{x}) \geq 0$ holds, i.e., the ego vehicle maintains a safe distance from the lead vehicle during the entire trip. 
The relative speed profile of the ego vehicle, the lead vehicle, and the variable speed limit $V_{max}(t)$ are summarized in Fig.~\ref{fig:Vr_vs_Time}, showing that $V_{f} \leq V_{max}(t)$ always holds. 
% i.e., $\mathcal{G}_{[0,\infty)]} h_v(t,\mathbf{x}) \geq 0$ holds.
%                                  
Finally, as shown in Fig.~\ref{fig:Xf_vs_Time}, the ego vehicle complies with the traffic signals. The height of the horizontal lines represents the relative position of the traffic signals with respect to the origin. 
% Distance between consecutive traffic signals varies. 
The colors represent the states of the traffic signals with respect to time. The  signals are not synchronized. The ego vehicle never crosses a traffic signal when the state is $\mathtt{Red}$. Moreover, since $h_{pos}(t,\mathbf{x})$ is always non-negative in Fig.~\ref{fig:hpos_vs_Time}, the traffic signal contract is satisfied. 
Fig.~\ref{fig:U_vs_Time} provides the control input $u$ as a function of time.

\section{Conclusions}\label{sec:conclusions}

We presented a compositional control synthesis method mapping a mission-level STL specification to an aggregation of contracts defined via continuously differentiable time-varying control barrier functions. The barrier functions are used to constrain the lower-level control synthesis problem, which is solved via quadratic programming. We illustrated the effectiveness of the proposed algorithm on a case study motivated by vehicular mission planning under safety constraints as well as constraints imposed by traffic regulations. Future work includes investigating extensions of the approach to more expressive fragments of STL as well as multi-agent planning. 

\bibliographystyle{IEEEtran}
\bibliography{biblio}

\end{document}